\documentclass[conference]{IEEEtran}
\ifCLASSINFOpdf
\else
\fi
\usepackage{amsfonts}
\usepackage{amsmath}
\usepackage{amsthm}
\usepackage{amssymb}
\usepackage{amsfonts}
\usepackage{amsmath}
\usepackage{amsthm}
\usepackage{graphicx}
\usepackage{epstopdf}
\newtheorem*{theo}{Theorem}
\newtheorem{lemma}{Lemma}

\newtheorem*{remark}{Remark}
\begin{document}
%%%
\title{Overloaded Multiuser MISO Transmission with Imperfect CSIT}
\author{\IEEEauthorblockN{Enrico~Piovano,
        Hamdi~Joudeh
        and~Bruno~Clerckx}
        \IEEEauthorblockA{Department of Electrical and Electronic Engineering, Imperial College London,
        	United Kingdom \\
        	 Email: \{e.piovano15, hamdi.joudeh10, b.clerckx\}@imperial.ac.uk}

\thanks{Department of Electrical and Electronic Engineering, Imperial College London,
 United Kingdom, e-mail: \{e.piovano15, hamdi.joudeh10, b.clerckx\}@imperial.ac.uk.}}% <-this % stops a space
%%%%%%%%%%%%%%%%%%%
\maketitle
%%%%%%%%%%%%%%%%%%%%
\begin{abstract}	
A required feature for the next generation of wireless communication networks will be the capability to serve simultaneously a large number of devices with heterogeneous CSIT qualities and demands.
In this paper, we consider the overloaded MISO BC with two groups of CSIT qualities. We propose a transmission scheme where degraded symbols are superimposed on top of spatially-multiplexed symbols.
The developed strategy allows to serve all users in a non-orthogonal manner and the analysis shows an enhanced perfomance compared to existing schemes. Moreover, optimality in a DoF sense is shown.
\footnote{This work has been partially supported by the EPSRC of UK, under grant EP/N015312/1.}
\end{abstract}
%%%%%%%%%%%%%%%%%%%%
%%%%%%%%%%
\IEEEpeerreviewmaketitle
%%%%%%%%%
\section{Introduction} \label{Introduction}
%%%%%%%%%%%%%
Exploiting the spatial dimension of the wireless channel through multiuser-multiantenna techniques has become an inevitable necessity to
meet the requirements of future wireless networks.
%%%
It is well established that achieving such spatial-multiplexing gains is highly dependent on the availability of accurate Channel State Information at the Transmitter (CSIT) \cite{Jindal2006,Caire2010}.
Since  highly accurate CSIT is not always guaranteed, initial studies and deployments strived to apply multiantenna schemes that assume perfect CSIT to scenarios with partial CSIT \cite{Clerckx2013}.
However, recent breakthroughs in the study of Degrees of Freedom (DoF) unveiled that such approach is fundamentally flawed as it fails to achieve the information-theoretic limits of the channels \cite{Yang2013,Davoodi2016}.
On the other hand, insights drawn from such fundamental works have proved very promising for the design of future wireless networks \cite{Clerckx2016}.

The ability to simultaneously support a tremendous number of devices with heterogeneous demands and capabilities is amongst the various features envisioned for future wireless networks \cite{Alliance2015}.
Hence, it is expected that many networks will operate in overloaded regimes, roughly described as scenarios where the number of messages exceeds the number of transmitting antennas.
One fundamental example is captured by the Single-Input-Single-Output (SISO) Broadcast Channel (BC), widely studied in literature.
However, insights drawn from such studies are deemed insufficient when considering multiple antennas, as the SISO BC is robust against CSIT inaccuracies due to its degraded nature.
On the other hand, the study of overloaded multiantenna channels is uncommon, e.g. works on the Multiple-Input-Single-Output (MISO) BC with imperfect CSIT consider a number of users less or equal to the number of transmitting antennas \cite{Jindal2006,Caire2010,Yang2013,Davoodi2016}.
\subsection{An overloaded MISO BC with heterogeneous partial CSIT} \label{Introduction_overloaded_MISO_BC}
In this work, we make progress towards understanding the fundamental limits of overloaded multiantenna networks with heterogeneous partial CSIT.
%%%
We consider a MISO BC comprising a transmitter equipped with $M$ antennas, and $K > M$ single-antenna receivers (or users) indexed by $\mathcal{K} = \{1,\ldots,K\}$.
As in \cite{Yang2013,Davoodi2016}, partial instantaneous CSIT is captured by allowing the $k$-th user's CSIT error variance to decay with the Signal to Noise Ratio (SNR) $P$ as $O(P^{-\alpha_{k}})$ for some exponent $\alpha_{k} \in [0,1]$ that represents the CSIT quality.
It is well understood that $\alpha_{k}=0$ and $\alpha_{k}=1$ correspond to no-CSIT and perfect CSIT in a DoF sense, respectively.
While a general  heterogeneous setup assumes arbitrary CSIT qualities, we restrict the analysis to the case where partial CSIT for $M$ of the $K$ users is available ($\alpha_{k} > 0$), while no-CSIT is available for the remaining $K-M$ users ($\alpha_{k} = 0$)\footnote{In this paper, no-CSIT implies that the transmitter has no (or finite precision \cite{Davoodi2016}) information about the channel direction. However, the channel gain (or long term SNR) is known to guarantee reliable communication.}.
We further simplify the analysis by considering a symmetric scenario where all users with partial CSIT have the same quality $\alpha$.
Such setup is sufficient to gain some insight into the structure of the DoF-optimum transmission scheme and the influence of heterogeneous partial CSIT.
Before we proceed, let us denote the groups of receivers by $\mathcal{K}_{\alpha}$ and $\mathcal{K}_{0}$, where the subscript indicates the CSIT quality.
\subsection{Time Partitioning and Power Partitioning}
In the presence of only one of the two groups $\mathcal{K}_{\alpha}$ and $\mathcal{K}_{0}$, DoF-optimum schemes are known.
In particular, the optimum sum-DoF for group $\mathcal{K}_{\alpha}$ is achieved through a Rate-Splitting (RS) scheme, which relies on the transmission of a degraded common symbol on the top of the classical Zero-Forced (ZF) private symbols \cite{Joudeh2016b}.
On the other hand, the absence of CSIT results in a collapse of the sum-DoF to unity \cite{Davoodi2016}, and the degraded layer becomes sufficient to achieve the DoF of group $\mathcal{K}_{0}$.
Hence, it is natural to think about serving each group independently through orthogonal time partitioning (or sharing).
Interestingly, we show that such strategy is in fact suboptimal in a DoF sense by proposing a superior strategy.
We propose a transmission scheme where the signals carrying the messages of groups $\mathcal{K}_{0}$ and $\mathcal{K}_{\alpha}$ are superimposed and separated in the power domain.
Users in $\mathcal{K}_{0}$ decode their symbols by treating the interference caused by the signals intended to $\mathcal{K}_{\alpha}$  as noise.
On the other hand, users in $\mathcal{K}_{\alpha}$ first decode the symbols intended to $\mathcal{K}_{0}$ (without hurting their DoF!), and then proceed to decode their own symbols.
Contrary to the orthogonal time partitioning, this leads to a non-orthogonal power partitioning.
First, we show that such strategy achieves a strict DoF gain over time partitioning when users in each group achieve a symmetric-DoF.
Second, we show that this strategy in fact achieves the optimum DoF region for the considered MISO BC.
Third, we show using simulations that the DoF gains achieved through power partitioning over time partitioning manifest in the finite SNR regime as significant achievable rate gains\footnote{\emph{Notation}: boldface lowercase, standard letters and calligraphic symbols denote column vectors, scalars and sets, respectively. The superscripts $(\cdot)^{T}$ and $(\cdot)^{H}$ denote the transpose and conjugate-transpose respectively. $\|\cdot\|$ and $\perp$ denote the Euclidian norm of a vector and orthogonality, respectively.}.
\section{System Model} \label{system_model}
At the $t$-th channel use of the considered MISO BC, the received signal at the $k$-th receiver is given by
\begin{equation} \label{eq:rxsignal}
y_k(t)=\mathbf{h}^H_k(t)\mathbf{x}(t)+n_k(t)
\end{equation}
where $\mathbf{h}_k^H(t) \in \mathbb{C}^{1 \times M}$  is the  channel vector and
$\mathbf{x}(t) \in \mathbb{C}^{M \times 1}$ is the transmitted signal, which is subject to the power constraint $\mathbb{E}(\| \mathbf{x}(t) \|^2) \leq P $. The term $n_k(t) \sim \mathcal{CN}(0,1)$ is the additive noise at the $k$-th receiver.
As described in Section \ref{Introduction_overloaded_MISO_BC}, the transmitter has access to an imperfect estimate of the instantaneous channel.
Denoting the estimate of the channel for the $k$-th user at the $t$-th channel use by $\hat{\mathbf{h}}_k(t)$, we have ${\mathbf{h}}_k(t)={\hat{\mathbf{h}}}_k(t)+\tilde{\mathbf{h}}_k(t)$, where $\tilde{\mathbf{h}}_k(t)$ is the channel estimation error at the transmitter.
The channel estimate ${\hat{\mathbf{h}}}_k(t)$ and the estimation error
${\tilde{\mathbf{h}}}_k(t)$  are assumed to be uncorrelated, with zero mean and
covariance matrices $(1-\sigma_k^2)\mathbf{I} $ and $\sigma_k^2 \mathbf{I}$, respectively, where $\sigma_k^2 \leq 1$.
For the sake of notational convenience, the channel user index
$t$ is omitted in the rest of the paper.
The CSIT error $\sigma_k^2$ decays with increasing SNR as $O(P^{-\alpha})$ for all $k \in \mathcal{K}_{\alpha}$, and $O(1)$ for all
$k \in \mathcal{K}_{0}$.
Moreover, without loss of generality, we assume that $\mathcal{K}_{\mathrm{\alpha}}=\{1,\ldots,M\}$ and
$\mathcal{K}_0=\{M+1,\ldots,K\}$.
The transmitter has messages $W_1,\ldots, W_K$ intended to the corresponding users.
Codebooks, probability of error, achievable rate tuples $(R_1(P),\ldots,R_K(P))$ and the capacity region $\mathcal{C}(P)$ are all defined in the Shannon theoretic sense.
The DoF tuple $(d_1, \ldots, d_K)$ is said to be achievable if there exists $(R_1(P),\ldots,R_K(P)) \in \mathcal{C}(P)$ such that $d_k=\lim_{P \to \infty} \frac{R_k(P)}{\log(P)}$ for all $k \in \mathcal{K}$.
The DoF region is defined as the closure of all achievable DoF tuples $(d_1, d_2, \dots, d_K)$, and is denoted by $\mathcal{D}$.
\section{A Time Partitioning approach} \label{problem_statement}
Since group $\mathcal{K}_{\alpha}$ has (partial) CSIT and group
$\mathcal{K}_0$ has no-CSIT, it seems natural to partition the time resource and carry out the transmission over two phases.
In particular, the first phase occupies a fraction $b \in [0,1]$ of the time in which group
$\mathcal{K}_{\alpha}$ is served using a multiuser scheme that leverages partial CSIT and achieves spatial-multiplexing gains.
On the other hand, the second phase occupies the remaining $1-b$ fraction of the time in which group $\mathcal{K}_0$ is served
with no multiplexing gains due to to the absence of CSIT.
This time partitioning scheme acts as a baseline for the scheme proposed in the following section.
Moreover, the two phases are in fact used as basic building blocks to construct the proposed scheme.
Next, we describe the two phases in more detail.
\subsubsection{Phase 1}
For the first phase where users with CSIT are served, we adopt the
RS strategy which is particularly suitable for scenarios with partial CSIT \cite{Hao2015,Joudeh2016b}.
In particular users $k \in \mathcal{K}_{\mathrm{\alpha}}$ split their respective messages into $\big(W_k^{(p)},W_k^{(c)}\big)$,
where $W_k^{(p)}$ is a private sub-message and $W_k^{(c)}$ is a common (or public) sub-message.
The sub-message $W_k^{(p)}$ is encoded into the private symbol $x_k^{(p)}$ decoded only by user $k$, while  $W_1^{(c)}, \ldots, W_M^{(c)}$ are jointly encoded into the common symbol $x^{(c)}$ decoded by all users in $\mathcal{K}_{\mathrm{\alpha}}$.
It is assumed that all symbols are drawn from Gaussian codebooks with unitary powers.
All symbols are linearly precoded and power allocated from which the transmitted signal is given by
%%%
\begin{equation} \label{tx_phase_1}
\mathbf{x}=\sqrt{P^{(c)}}\mathbf{v}^{(c)}x^{(c)}+\sum_{k \in \mathcal{K}_{\mathrm{\alpha}}}{\sqrt{P_k^{(p)}}\mathbf{v}_k^{(p)}x_k^{(p)}}
\end{equation}
%%%
where $\mathbf{v}^{(c)} \in \mathbb{C}^{M \times 1}$ and $\mathbf{v}_k^{(p)} \in \mathbb{C}^{M \times 1}$ are unitary precoding vectors, and $P^{(c)}$ and $P_k^{(p)}$ are the corresponding allocated powers with $P^{(c)}+\sum_{k \in \mathcal{K}_{\mathrm{\alpha}}}{P_k^{(p)}} \leq P$.
Since the common symbol is decoded by all users, $\mathbf{v}^{(c)}$ is chosen as a random (or generic) precoding vector.
On the other hand, the private symbols are precoded by ZF over the channel estimate, i.e.
$\mathbf{v}_k^{(p)} \perp \big\{\hat{\mathbf{h}}_l \big\}_{l \in \mathcal{K}_{\mathrm{\alpha}} \setminus k}$.
The power allocation is set such that $P^{(c)} = O(P)$ and $P_k^{(p)} = O(P^{\alpha})$.

All users decode the common symbol by treating the interference from all private symbols as noise, from which the Signal to Interference plus Noise Ratio (SINR) scales as $O(P^{1-\alpha})$.
This is followed by removing the common symbol, and then each receiver decodes its private symbol with SINR of $O(P^{\alpha})$.
Normalized by the time partitioning factor $b$, the DoF achieved by the common symbol is given by $1- \alpha$, while each private symbol achieves a DoF of $\alpha$ \cite{Joudeh2016b}.
Hence, the  per user symmetric normalized DoF achieved by evenly sharing the common symbol is given by $\frac{1+(M-1)\alpha}{M}$.
%%%
\subsubsection{Phase 2}
In the second phase, users $k \in \mathcal{K}_{\mathrm{0}}$ are served.
Since all users have no-CSIT, after normalizing by the time partition $1-b$, the sum-DoF collapses to 1 \cite{Davoodi2016}.
This single normalized DoF can be shared in an orthogonal fashion using time-sharing or in a non-orthogonal fashion using superposition coding and Successive Interference Cancelation (SIC).
From a DoF perspective, these two strategies achieve the same performance.
%%%
Assuming superposition coding, messages are encoded into symbols and then precoded such that
%%%
\begin{equation} \label{tx_phase_2}
\mathbf{x}=\sum_{k \in \mathcal{K}_{\mathrm{0}}}{\sqrt{P_k} \mathbf{v}_k x_k}
\end{equation}
where $x_{k}$ is an encoded symbol, $\mathbf{v}_k$ is a random unitary precoding vector and $P_{k}$ is the power allocation.

%%%
Using an appropriate power allocation, it can be shown that the single normalized DoF can be split evenly amongst users such that each user achieves a normalized DoF of $\frac{1}{K-M}$.
%%%
\subsubsection{Achievable DoF}
%%%
It can be seen that within each phase (or group), power allocation is carried out such that users achieve symmetric normalized DoF.
By incorporating the time partitioning factor $b \in [0,1]$, the actual (non-normalized) DoF achieved by the $k$-th user is given by
\begin{equation} \label{dof_scheduling}
d_k=\begin{cases}
b\frac{1+(M-1)\alpha}{M},  &  k \in \mathcal{K}_{\mathrm{\alpha}}\\
(1-b)\frac{1}{K-M},  &  k \in \mathcal{K}_{\mathrm{0}}.
\end{cases}
\end{equation}
The time partitioning factor can be further optimized to achieve a symmetric-DoF amongst all users in the system, or any other tradeoff depending on the design objective.
%%%
\section{A Power Partitioning Approach} \label{proposed scheme}
%%%
In contrast to the time partitioning approach in the previous section, we propose a scheme based on power partitioning, also known as signal-space partitioning \cite{Yuan2016}.
For some partitioning factor $\beta \in [0,1]$, the bottom $\beta$ power levels are reserved for the transmission to $\mathcal{K}_{\alpha}$ with partial CSIT, while the top $1- \beta$ power levels are occupied by the transmission to $\mathcal{K}_{0}$ with no-CSIT.
It can be seen that power partition $\beta$ in this scheme is reminiscent to the time partition $b$ in the previous scheme.
%%%
Moreover, the transmitted signal is in fact a superposition of the signals in (\ref{tx_phase_1}) and (\ref{tx_phase_2}) such that
\begin{equation} \label{TX_signal_PS}
\begin{split}
\mathbf{x}=&\sqrt{P_0}\sum_{i \in \mathcal{K}_{\mathrm{0}}}{\sqrt{q_i} \mathbf{v}_i x_i}+ \sqrt{P^{(c)}}\mathbf{v}^{(c)}x^{(c)} \\
& +\sum_{k \in \mathcal{K}_{\mathrm{\alpha}}}{\sqrt{P_k^{(p)}}\mathbf{v}_k^{(p)}x_k^{(p)}}
\end{split}
\end{equation}
where symbols, precoding vectors and powers are as defined in the previous section.
To highlight the power partitioning, we introduce $P_0$ which denotes the total power allocated to the signal intended to all users in $\mathcal{K}_0$.
It follows that $q_i=P_i/P_0$ is the normalized power allocated user $i \in \mathcal{K}_0$.
An example that illustrates the two scheme is given in Fig. \ref{fig:figura_1}.
\begin{figure}[]
	\centering
	\includegraphics[width=0.48\textwidth]{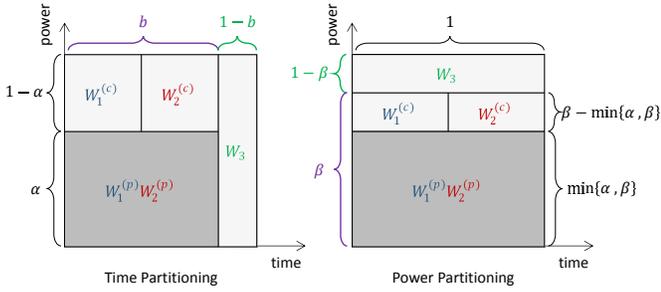}
	\caption{Time partitioning and power partitioning for $M=2$ and $K=3$. Define the normalized spatial-multiplexing as the sum-DoF normalized by both the time partition and power partition. The normalized spatial-multiplexing gain in rectangles with light and dark shadings is $1$ and $2$ respectively.}
	\label{fig:figura_1}
\end{figure}

Before we proceed to take a closer look at the power partitioning, it is useful to highlight that
the signal received by the $k$ users is expressed by
\begin{equation} \label{rx_our_alpha}
\begin{split}
y_k=& \sqrt{P_0}\sum_{i \in \mathcal{K}_{\mathrm{0}}}{{\sqrt{q_i}\mathbf{h}^H_k \mathbf{v}_i x_i}}+
 \sqrt{P^{(c)}} \mathbf{h}_k^H\mathbf{v}^{(c)}x^{(c)} \\ &
+ \sum_{i \in \mathcal{K}_{\mathrm{\alpha}}}{\sqrt{P_i^{(p)}}}\mathbf{h}^H_k \mathbf{v}_i^{(p)}x_i^{(p)} + {n_k},
\end{split}
\end{equation}
in which all different desired and interference components can be seen.
%%%
In order to partition the signal-space through the power domain, the power allocation is carried out such that
\begin{equation}
\begin{cases}
P_0=O(P) \\
P^{(c)} + \sum_{k \in \mathcal{K}_{\alpha}} P_{k}^{(p)} = O(P^{\beta}).
\end{cases}
\end{equation}
\setcounter{subsubsection}{0}
\subsubsection{No-CSIT Receivers}
Users in  $\mathcal{K}_{0}$ decode their messages by treating the interference (consisting of signals intended to users in $\mathcal{K}_{\alpha}$) as noise.
This is equivalent to raising the noise floor to $P^{\beta}$ in Phase 2 of the previous section.
Hence, the sum-DoF achieved by users in $\mathcal{K}_{0}$ is given by $1-\beta$.
Through an appropriate allocation of $\{q_{i}\}_{i\in\mathcal{K}_{0}}$, this DoF can be split evenly amongst users in $\mathcal{K}_{0}$.
\subsubsection{Partial CSIT Receivers}
As for users in $\mathcal{K}_{\alpha}$, the same RS strategy of Phase 1 in the previous section is carried out where the power of $O(P^{\beta})$ is further split between the common symbol and the private symbols. In particular, the common symbol is allocated a power of $O(P^{\beta})$, while private symbols are allocated a power of $O(P^{a})$ where $a \leq \beta$.
%%%
Before decoding their symbols, receivers first decode all symbols intended to users in $\mathcal{K}_{0}$ and remove them from the received signal.
Since such messages are degraded already, the DoF achieved by users in $\mathcal{K}_{0}$  remains uninfluenced by this step.
On the other hand, users in $\mathcal{K}_{\alpha}$ now fully occupy the bottom $\beta$ power levels.

Users in $\mathcal{K}_{\alpha}$ now proceed to decode the common symbol as in Phase 1 of the previous section.
This is received with a SINR of $O(P^{\beta-a})$, hence achieves a DoF of $\beta-a$.
After removing the common symbol, each receiver decodes its private symbol with SINR of $O(P^{a})$, achieving a DoF of $a$.
It remains to highlight that since the channel estimation error scales as $O(P^{-\alpha})$, and due to ZF, each receiver in $\mathcal{K}_{\alpha}$  experiences an interference from the other private symbols that scales as $O(P^{a-\alpha})$.
This is drowned by noise if $a \leq \alpha$.
Knowing that $a \leq \beta$, we may set $a=\min\{\alpha,\beta\}$.
In other words, as long as the partition $\beta$ satisfies  $\beta \leq \alpha$, users in $\mathcal{K}_{\alpha}$ only need to rely on private messages using ZF as interference can be drown by noise and RS is unnecessary.
For partitions with $\beta > \alpha$, ZF is insufficient to neutralize interference, and RS becomes useful for users in $\mathcal{K}_{\alpha}$.
It follows that each private symbol achieves a DoF of $\min\{\alpha,\beta\}$, while the common symbol achieves a DoF of $\beta - \min\{\alpha,\beta\}$.
%%%
\begin{remark}
\nonumber
The proposed scheme is a superposition of non-orthogonal layers and an orthogonal layer. Non-orthogonal layers consist of degraded symbols coming from $\mathcal{K}_{0}$ and RS in $\mathcal{K}_{\alpha}$, decoded by treating the orthogonal-layer as noise, and removed using SIC.
The orthogonal layer consist of spatially-multiplexed symbols carrying the remaining information for $\mathcal{K}_{\alpha}$,  which see no interference due to SIC of non-orthogonal layers and ZF up to the $\alpha$-th power level.
\end{remark}
%%%
\subsubsection{Achievable DoF}
%%%
As in the previous section, we consider the case where users in each group achieve a symmetric-DoF.
It follows that the DoF achieved by the $k$th user is given by
\begin{equation} \label{dof_PS}
d_k=\begin{cases}
\frac{\beta+(M-1)\cdot \min\{\alpha,\beta\}}{M},  &  k \in \mathcal{K}_{\mathrm{\alpha}}\\
(1-\beta)\frac{1}{K-M},  &  k \in \mathcal{K}_{\mathrm{0}}.
\end{cases}
\end{equation}
Moreover, $\beta$ can be optimized to achieve different tradeoffs.
\subsubsection{Gain over time partitioning}
%%%
Here we demonstrate that the power partitioning scheme achieves a DoF gain over the time partitioning scheme.
Let $d_k^{\mathrm{(tp)}}$ be the DoF achieved by the $k$th user through time partitioning as in the previous section, i.e. obtained using (\ref{dof_scheduling}) for some partition $b$.
To highlight the DoF gains, let us consider the symmetric-DoF achieved by users in $\mathcal{K}_{\mathrm{\alpha}}$ through power partitioning given that users in $\mathcal{K}_{\mathrm{0}}$  maintain the same DoF as in time partitioning, i.e. $d_{k} = d_k^{\mathrm{(tp)}}$ for all $k \in \mathcal{K}_{\mathrm{0}}$.
To achieve this, we need to set $\beta=b$ in the power partitioning scheme.
It follows from (\ref{dof_PS}) that the DoF of the remaining users is given by
\begin{equation}
d_k=\begin{cases} \label{eq_pre_1}
\frac{b+(M-1)\alpha}{M}, & \alpha \leq b\\
b, & \alpha > b
\end{cases}
\text{ , for all }k \in \mathcal{K}_{\mathrm{\alpha}}.
\end{equation}
It can be seen that $d_k \geq d_k^{\mathrm{(tp)}}$ for all $k \in \mathcal{K}$.
%%%
For $k \in \mathcal{K}_{0}$, this follows directly from the design criteria.
For the remaining users $k \in \mathcal{K}_{\alpha}$, this follows by noting that for all  $\alpha,b \leq 1$, we have $\frac{b+(M-1)\cdot \min \{\alpha,b\}}{M} \geq \frac{b+(M-1)\cdot \alpha b}{M}$.
Moreover, this inequality is strict whenever $0<\alpha,b < 1$, i.e. partial CSIT for $\mathcal{K}_{\alpha}$ and non-zero (or unity) partitioning.
Under such conditions, power partitioning achieves a strict improvement in the DoF of users in $\mathcal{K}_{\alpha}$ over time partitioning.

To gain more insight into the DoF gain, consider the example shown in Fig. \ref{fig:figura_1}.
It can be seen that the DoF achieved in each rectangle (a time-power resource block) is given by the rectangle's area times the normalized spatial-multiplexing gain (2 for ZF and 1 for degraded).
First, assume that user-3 is switched off. The sum-DoF achieved by the remaining two users through RS is given by $1+\alpha$.
Now, introducing user-3 through time partitioning reduces the sum-DoF to $b(1+\alpha) + (1-b) = 1+b\alpha$.
On the other hand, user-3 is introduced through power partitioning without harming the sum-DoF as long as $\beta = b \geq \alpha$.
Keeping in mind that user-3 achieves the same DoF in both cases, it follows that user-1 and user-2 achieve higher DoF in the latter.
For $\beta = b < \alpha$, introducing user-3 through power partitioning reduces the sum-DoF to $1+b$.
However, this is still higher than the sum-DoF of $1+b\alpha$ achieved through time partitioning.
%%%
\section{Optimum DoF Region} \label{DoF region}
%%%
In the previous section, we considered the case where DoF tuples of the form $(d_{\alpha},\ldots,d_{\alpha},d_{0},\ldots,d_{0})$ are achieved, i.e. users in $\mathcal{K}_{\alpha}$ achieve the symmetric-DoF of $d_{\alpha}$ while users in $\mathcal{K}_{0}$ achieve the symmetric-DoF of $d_{0}$.
This gave some insight into the gains achieved through power partitioning as opposed to time partitioning.
However, in more general scenarios, achievable DoF tuples assume a wide variety of tradeoffs characterized by achievable and optimum DoF regions.
Interestingly, the optimum DoF region for the considered setup is achieved through variants of the power partitioning scheme proposed in the previous section.
This region is characterized in the following result.
\begin{theo} \label{theo_proposed_scheme}
For the overloaded MISO BC described in Sections \ref{Introduction_overloaded_MISO_BC} and \ref{system_model}, the optimum DoF region $\mathcal{D}$ is given by
\begin{equation} \label{dis_1}
  	 d_k \geq 0, \quad \forall k \in \mathcal{K}
  	\end{equation}
  	\begin{equation} \label{dis_2}
  	\sum_{k \in \mathcal{S}}{d_{k}} +  \sum_{k \in \mathcal{K}_{\mathrm{0}}  }{d_k}  \leq 1 + (|\mathcal{S}|-1) \alpha, \quad \forall \mathcal{S} \subseteq  \mathcal{K}_{\mathrm{\alpha}}, |\mathcal{S}| \geq 1.
  	\end{equation}
\end{theo}
The achievability of the DoF region is based on generalizing the power partitioning scheme of Section \ref{proposed scheme} by allowing arbitrary power allocations and splits of the common message. On the other hand, the converse is based on the sum-DoF upperbound in \cite{Davoodi2016}.
The complete proof is given in the Appendix.
\begin{figure}[]
	\centering
	\includegraphics[trim=0mm 0mm 0mm 0mm,clip=true,width=0.50\textwidth]{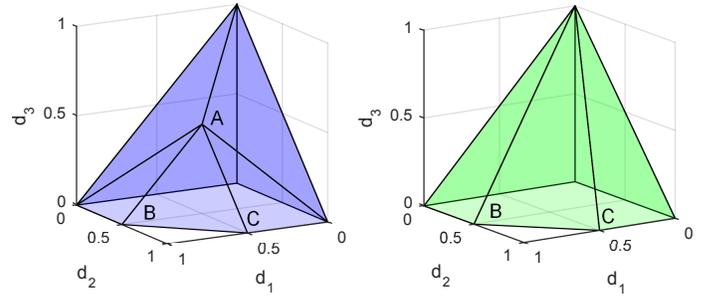}
	\caption{DoF region achieved by power partitioning (left) and time partitioning (right) for $M=2$ and $K=3$, and CSIT quality $\alpha=0.5$ for the first two users. The points are $A=(\alpha,\alpha,1-\alpha)$, $B=(1,\alpha,0)$ and $C=(\alpha,1,0)$. It can be seen that $A$ cannot be achieved through time partitioning.}
	\label{fig:figura_2}
\end{figure}

Note that from (\ref{dis_2}), we have $d_k \leq 1$ for all $k \in \mathcal{K}$ which is a trivial upperbound for the per-user DoF,
and $\sum_{k \in \mathcal{K}_{\mathrm{0}}}{d_k} \leq 1$ which limits the sum-DoF of the no-CSIT users to unity.
To better visualize the optimum DoF region, an example is given in Fig. \ref{fig:figura_2} (left) for a channel with $M=2$ and $K=3$, where the CSIT quality of the first two users is $\alpha=0.5$.
Moreover, for the sake of comparison, the DoF region achieved through time partitioning is shown in Fig. \ref{fig:figura_2} (right).
The time partitioning region is obtained by time-sharing the DoF of 1 achieved by user-3 with the DoF region of the two remaining users achieved through RS (see \cite{Clerckx2016}).
For the power partitioning region, the facet given by $A-B-C$ is in fact sum-DoF optimum. Hence, user-3 can be served with non-zero DoF without influencing the Sum-DoF (e.g. point $A$). On the other hand, serving user-3 with non-zero DoF through time partitioning is not possible without decreasing the sum-DoF as it requires moving away from the segment $B-C$.
\section{Numerical Results} \label{Simulation Results}
In this section, we show that the obtained DoF gains translate into enhanced rate performances.
We consider a MU-MISO scenario with $M=2$ antennas and $K=3$ users.
Uncorrelated channels are assumed with entries drawn from $\mathcal{CN}(0,1)$. 
Users 1 and 2 have CSIT qualities $\alpha$, where channel estimation errors
have entries drawn from
$\mathcal{CN}(0,\sigma^2)$ with $\sigma^2 = P^{-\alpha}$.
On the other hand, the instantaneous CSIT of user 3 is unknown.
We numerically evaluate the ergodic sum rate of the first two users
achieved by power partitioning and time partitioning, while maintaining the ergodic rate of the third user to be the same in both cases.
This is obtained by properly tuning, in the power partitioning approach, the power $P_0$ allocated to the symbol of the third user, while considering a RS strategy for the first two users.
\begin{figure}[]
	\centering
	\includegraphics[width=0.50\textwidth]{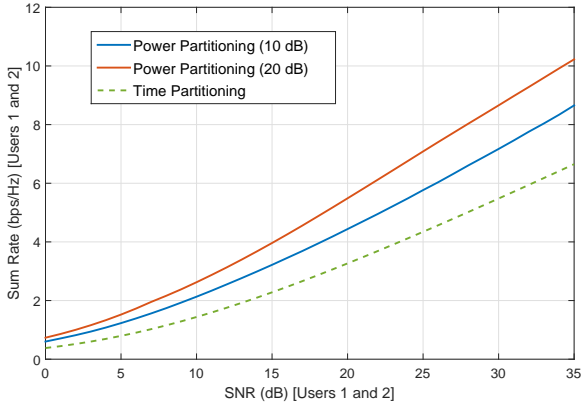}
	\caption{Sum rate of user-1 and user-2 while maintaining the same rate for user-3 for the cases when the long-term
		SNR of user-3 is $10$ dB and $20$ dB lower that user-1,2. The parameters taken are $\alpha=0.5$ and $b=0.5$.}
	\label{fig:Plot_1}
\end{figure}
Fig. \ref{fig:Plot_1} shows the sum rate of user 1 and user 2 with respect to their long-term SNR for both power partitioning and time partitioning. We assume a scenario with $\alpha=0.5$ and we set the parameter $b=0.5$.
We consider two cases where the SNR of users 1 and 2 is taken to be $10$ dB, and then $20$ dB, larger than the SNR of user 3.
Since in time partitioning users 1 and 2 are scheduled separately from user 3, the difference in SNR only affects their sum rate performance in power partitioning. In the legend, we include this difference by brackets (only for power partitioning).
From Fig. \ref{fig:Plot_1}, it is evident that our proposed power partitioning based approach significantly outperforms time partitioning in both cases.
Furthermore, as the difference between the SNR of users 1 and 2 and the SNR of user 3 becomes larger, the rate gain increases which cannot be seen from DoF analysis.
\section{Conclusion} \label{conclusion}
In this paper, we considered an overloaded MISO BC where the transmitter has partial CSI for $M$ users (equal to the number of antennas) and no-CSI for the remaining $K-M$.
We proposed a transmission scheme based on power partitioning and showed that it achieves strict DoF gains
compared to a scheme where the two sets of users are independently served over orthogonal time slots.
Moreover, we showed that  the optimum DoF region for such channel is in fact achieved by generalizing the proposed power partitioning scheme.
The finite SNR rate performance of the proposed DoF-motivated scheme is evaluated through simulations in which significant gains over time partitioning are demonstrated.
This shows that such DoF-motivated design and analysis can be indeed very useful in guiding the design and optimization of more efficient practical transmission strategies.

%\appendices
\section*{Appendix}

\section*{Proof of the optimum DoF region $\mathcal{D}$}
The DoF region $\mathcal{D}$ described by the inequalities in (\ref{dis_1}) and (\ref{dis_2}) is a $K$ dimensional polyhedron.
To prove the optimality of $\mathcal{D}$ we show that it is both achievable, and an outer bound of the optimal region.

\textit{Achievability}:
%%%
In this section we prove the achievability of $\mathcal{D}$.
Before we delve into the general case, we first characterize the achievable DoF region obtained by switching off all users in $\mathcal{K}_0$ (forcing their DoF to zero).
This is equivalent to projecting $\mathcal{D}$ onto the $M$ dimensional subspace characterized by $d_{M+1}, \ldots, d_K=0$.
This region is then utilized as a building block to prove the achievability of $\mathcal{D}$.
\begin{lemma} \label{prop_M=K}
	For a MISO BC with $K=M$ and CSIT quality $\alpha \in [0,1]$ for all users, an achievable DoF region $\mathcal{D}_{M=K}$ is given by
	\begin{equation} \label{dis_1_M=K}
	d_k \geq 0, \quad \forall k \in \mathcal{K}
	\end{equation}
	\begin{equation} \label{dis_2_M=K}
	\sum_{k \in \mathcal{S}}{d_{k}}  \leq 1 + (|\mathcal{S}|-1) \alpha, \quad \forall \mathcal{S} \subseteq  \mathcal{K}, |\mathcal{S}| \geq 1
	\end{equation}
	where $\mathcal{K}$ denotes the set of users $\{1,\ldots,K\}$.
\end{lemma}
\begin{proof}
	The region $\mathcal{D}_{M=K}$ is a polyhedron given by the intersection of the half-planes delimited
	by the hyperplanes in (\ref{dis_1_M=K}) and (\ref{dis_2_M=K}).
	We show that $\mathcal{D}_{M=K}$ is achievable by induction over the number of users $K$. The hyphothesis is clearly true for $K=1$. We assume that the hyphothesis
	is valid for $K=1,\ldots,k-1$.
	First, each of the hyperplanes in (\ref{dis_1_M=K}) and (\ref{dis_2_M=K}) contains a facet of the polyhedron and the set of all the facets corresponds to the boundary of $\mathcal{D}_{M=K}$.
	We start by considering the hyperplanes in (\ref{dis_2_M=K}) and the corresponding facets.
	For each subset $\mathcal{S} \subseteq \mathcal{K}$ with $|\mathcal{S}| \geq 1$, we need to show that the corresponding facet, given by all the non-negative tuples $(d_1,\ldots,d_k)$ which satisfy
	\begin{equation*}
	\begin{cases}
	\sum_{i \in \mathcal{S}}{d_{i}}  = 1 + (|\mathcal{S}|-1) \alpha\\
	\sum_{i \in \bar{\mathcal{S}}}{d_{i}}  \leq 1 + (|\bar{\mathcal{S}|}-1) \alpha, &  \forall {\bar{\mathcal{S}}} \subseteq {\mathcal{K}}, \; \bar{\mathcal{S}} \neq \mathcal{S}, \; |\bar{\mathcal{S}}| >1
	\end{cases}
	\end{equation*}
	is achievable.

    Considering $j \in \mathcal{K} \setminus {\mathcal{S}}$, the two conditions $\sum_{i \in \mathcal{S}}{d_{i}}  = 1 + (|\mathcal{S}|-1) \alpha$  and $\sum_{i \in \mathcal{S} \cup \{j\}}{d_{i}}  \leq 1 + |\mathcal{S}|\alpha$ must be satisfied, hence $d_j \leq \alpha$.
	On the other hand, for $j \in \mathcal{S}$, we examine two cases. In case of $|\mathcal{S}|=1$, we have $d_j=1$. In case of $\mathcal{S} \geq 2$, the conditions $\sum_{i \in \mathcal{S}}{d_{i}}  = 1 + (|\mathcal{S}|-1) \alpha$ and $\sum_{i \in \mathcal{S} \setminus \{j\}}{d_{i}}  \leq 1 + (|\mathcal{S}|-2) \alpha$ must be simultaneously satisfied, hence $d_j \geq \alpha$.

	From the above analysis, the conditions on the set of tuples $(d_1,\ldots,d_k)$ can be equivalently written
	as
	\begin{equation*}
	\begin{cases}
	\sum_{i \in \mathcal{S}}{d_{i}}  = 1 + (|\mathcal{S}|-1) \alpha\\
	d_i \geq \alpha, & \forall i \in \mathcal{S}\\
	d_i \leq \alpha, & \forall i  \in \mathcal{K} \setminus \mathcal{S}.
	\end{cases}
	\end{equation*}
	Each DoF tuple is achieved through RS by allocating powers
	scaling as $O(P^{\alpha})$ to private symbols of users $i \in \mathcal{S}$, and
	powers scaling as $O(P^{d_{i}})$  to private symbols of users $i \in \mathcal{K} \setminus \mathcal{S}$. The common symbol's DoF is split, in all possible variants, among users $i \in \mathcal{S}$ only.
	
	We consider now the facets contained in the hyperplanes in (\ref{dis_1_M=K}).
	Taking any user $j \in \mathcal{K}$, a facet is given by all non-negative tuples $(d_1,\ldots,d_k)$ which satisfy
	\begin{equation*}
	\begin{cases}
	d_j=0 \\
	\sum_{i \in \mathcal{S}}{d_{i}}  \leq 1 + (|\mathcal{S}|-1) \alpha, & \forall \mathcal{S} \subseteq \mathcal{K} \setminus \{j\}, \; |\mathcal{S}| \geq 1.
	\end{cases}
	\end{equation*}
	This corresponds to the region in (\ref{dis_1_M=K}) and (\ref{dis_2_M=K}) when considering the $k-1$ users $\mathcal{K} \setminus \{j\}$. In this case we have $k$ antennas and $k-1$ users. However, increasing the number of antennas does not harm the achievable DoF region, hence the above region is achievable by induction.
	Since all facets of the polyhedron are achievable,
	all the remaining points can be achieved by time-sharing.
\end{proof}

We can now proceed to show the achievability of the region $\mathcal{D}$.
First, defining $d_{\Sigma}=\sum_{i \in \mathcal{K}_{0}}{d_i}$, the problem is equivalent to showing that all the
non-negative tuples $(d_1,\ldots,d_M,d_{\Sigma})$ that satisfy
\begin{equation} \label{dis_1_proof}
d_i \geq 0, d_{\Sigma} \geq 0 \quad \forall i \in \mathcal{K}_{\alpha}
\end{equation}
\begin{equation} \label{dis_2_proof}
\sum_{i \in \mathcal{S}}{d_{i}} +  d_{\Sigma}  \leq 1 + (|\mathcal{S}|-1) \alpha, \quad \forall \mathcal{S} \subseteq  \mathcal{K}_{\mathrm{\alpha}}, |\mathcal{S}| \geq 1
\end{equation}
are achievable. All tuples $(d_1,\ldots,d_M,d_{M+1},\ldots,d_K)$ are then obtained by splitting, in all
possible variants, the values of $d_{\Sigma}$ among users in $\mathcal{K}_0$.
The proof follows the same steps as before but, in this case, the induction is done over the number of users in $\mathcal{K}_{\mathrm{\alpha}}$, denoted as $K_{\alpha}$ and equal to $M$. The case $K_{\alpha}=1$ is trivial. We assume that the hypothesis holds for
$K_{\alpha}=1,\ldots,k-1$. As before, we show that each facet of the polyhedron is achievable.
Starting with the hyperplanes in (\ref{dis_2_proof}), for each subset $\mathcal{S} \subseteq \mathcal{K}_{\alpha}, |\mathcal{S}| \geq 1$,
we need to show that all the non-negative tuples $(d_1,\ldots,d_k,d_{\Sigma})$ that satisfy
\begin{equation*}
\begin{cases}
\sum_{i \in \mathcal{S}}{d_{i}} + d_{\Sigma} = 1 + (|\mathcal{S}|-1) \alpha\\
\sum_{i \in \bar{\mathcal{S}}}{d_{i}} + d_{\Sigma} \leq 1 + (|\bar{\mathcal{S}|}-1) \alpha, \forall {\bar{\mathcal{S}}} \subseteq {\mathcal{K}}_{\alpha}, \bar{\mathcal{S}} \neq \mathcal{S}, |\bar{\mathcal{S}}| \geq 1
\end{cases}
\end{equation*}
are achievable.
Following the same steps as before, it can be verified that the above conditions are equivalent to
\begin{equation*}
\begin{cases}
\sum_{i \in \mathcal{S}}{d_{i}} + d_{\Sigma} = 1 + (|\mathcal{S}|-1) \alpha\\
d_i \geq \alpha, & \forall i \in \mathcal{S}\\
d_i \leq \alpha, & \forall i  \in \mathcal{K}_{\mathrm{\alpha}} \setminus \mathcal{S}.
\end{cases}
\end{equation*}
Each DoF tuple is achieved through power partitioning by allocating powers
scaling as $O(P^{\alpha})$ to private symbols of users $i \in \mathcal{S}$, and
powers scaling as $O(P^{d_{i}})$  to private symbols of users $i \in \mathcal{K}_{\mathrm{\alpha}} \setminus \mathcal{S}$.
On top, we consider all possible power partitions $\beta \in [\alpha,1]$ and for each partition, the common symbol's DoF is split, in all possible variants, among users $k \in \mathcal{S}$ only, while $d_{\Sigma}=1-\beta$.

Considering the facets contained in the hyperplanes in (\ref{dis_2_proof}),
we have two cases.
The first is given by $d_{\Sigma}=0$ and it reduces to $k$ users with CSIT $\alpha$
and $k$ antennas as in Lemma \ref{prop_M=K}. The second case considers any $j \in \mathcal{K}_{\alpha}$ and we have
\begin{equation*}
\begin{cases}
d_j=0\\
\sum_{i \in \mathcal{S}}{d_{i}} + d_{\Sigma}  \leq 1 + (|\mathcal{S}|-1) \alpha, \forall \mathcal{S}\subseteq \mathcal{K}_{\mathrm{\alpha}} \setminus \{j\}, |\mathcal{S}| \geq 1. \\
\end{cases}
\end{equation*}
This corresponds to the region in (\ref{dis_1_proof}) and (\ref{dis_2_proof}) considering the $k-1$ users in $\mathcal{K}_{\alpha}$. Using the same argument as before, this region is achievable by induction. Moreover, all facets of the polyhedron are achievable, all the remaining points can be achieved by time-sharing

\textit{Converse}:
The converse is based on the sum-DoF upperbound obtained in \cite{Davoodi2016}.
For an arbitrary subset of users $\mathcal{U} \subseteq \mathcal{K}$, the sum-DoF is upperbounded by
\begin{equation}
\label{eq_sum_DoF_UB}
\sum_{k \in \mathcal{U} }{d_k} \leq 1 + \alpha(|\mathcal{S}|-1)^+
\end{equation}
where $\mathcal{S}=\mathcal{U} \cap \mathcal{K}_{\mathrm{\alpha}}$.
We increase the number of transmitter antennas to $K$ and then enhance the quality of one of the users in $\mathcal{S}$ to $1$ (if $\mathcal{S}$ is empty we pick any other user).
Since the previous steps provide an outerbound and cannot harm the DoF, \eqref{eq_sum_DoF_UB} directly follows from \cite[Theorem 1]{Davoodi2016}.
By removing all redundant inequalities, the outerbound coincides with the region $\mathcal{D}$.

\ifCLASSOPTIONcaptionsoff
  \newpage
\fi

\bibliographystyle{IEEEtran}
\bibliography{References}

% Generated by IEEEtran.bst, version: 1.12 (2007/01/11)
\begin{thebibliography}{10}
\providecommand{\url}[1]{#1}
\csname url@samestyle\endcsname
\providecommand{\newblock}{\relax}
\providecommand{\bibinfo}[2]{#2}
\providecommand{\BIBentrySTDinterwordspacing}{\spaceskip=0pt\relax}
\providecommand{\BIBentryALTinterwordstretchfactor}{4}
\providecommand{\BIBentryALTinterwordspacing}{\spaceskip=\fontdimen2\font plus
\BIBentryALTinterwordstretchfactor\fontdimen3\font minus
  \fontdimen4\font\relax}
\providecommand{\BIBforeignlanguage}[2]{{%
\expandafter\ifx\csname l@#1\endcsname\relax
\typeout{** WARNING: IEEEtran.bst: No hyphenation pattern has been}%
\typeout{** loaded for the language `#1'. Using the pattern for}%
\typeout{** the default language instead.}%
\else
\language=\csname l@#1\endcsname
\fi
#2}}
\providecommand{\BIBdecl}{\relax}
\BIBdecl

\bibitem{Jindal2006}
N.~Jindal, ``{MIMO} broadcast channels with finite-rate feedback,'' \emph{IEEE
  Trans. Inf. Theory}, vol.~52, no.~11, pp. 5045--5060, Nov 2006.

\bibitem{Caire2010}
G.~Caire, N.~Jindal, M.~Kobayashi, and N.~Ravindran, ``Multiuser {MIMO}
  achievable rates with downlink training and channel state feedback,''
  \emph{IEEE Trans. Inf. Theory}, vol.~56, no.~6, pp. 2845--2866, Jun 2010.

\bibitem{Clerckx2013}
B.~Clerckx and C.~Oestges, \emph{MIMO Wireless Networks: Channels, Techniques
  and Standards for Multi-antenna, Multi-user and Multi-cell Systems}.\hskip
  1em plus 0.5em minus 0.4em\relax Academic Press, 2013.

\bibitem{Yang2013}
S.~Yang, M.~Kobayashi, D.~Gesbert, and X.~Yi, ``Degrees of freedom of time
  correlated {MISO} broadcast channel with delayed {CSIT},'' \emph{IEEE Trans.
  Inf. Theory}, vol.~59, no.~1, pp. 315--328, 2013.

\bibitem{Davoodi2016}
A.~G. Davoodi and S.~A. Jafar, ``Aligned image sets under channel uncertainty:
  Settling conjectures on the collapse of degrees of freedom under finite
  precision {CSIT},'' \emph{IEEE Trans. Inf. Theory}, vol.~62, no.~10, pp.
  5603--5618, Oct 2016.

\bibitem{Clerckx2016}
B.~Clerckx, H.~Joudeh, C.~Hao, M.~Dai, and B.~Rassouli, ``{Rate splitting for
  MIMO wireless networks: a promising PHY-layer strategy for LTE evolution},''
  \emph{IEEE Commun. Magazine}, vol.~54, no.~5, pp. 98--105, May 2016.

\bibitem{Alliance2015}
{NGMN Alliance}, ``{5G white paper},'' Feb. 2015.

\bibitem{Joudeh2016b}
H.~Joudeh and B.~Clerckx, ``Sum-rate maximization for linearly precoded
  downlink multiuser {MISO} systems with partial {CSIT}: A rate-splitting
  approach,'' \emph{IEEE Trans. Commun.}, vol.~64, no.~11, pp. 4847--4861, Nov.
  2016.

\bibitem{Hao2015}
C.~Hao, Y.~Wu, and B.~Clerckx, ``Rate analysis of two-receiver {MISO} broadcast
  channel with finite rate feedback: A rate-splitting approach,'' \emph{IEEE
  Trans. Commun.}, vol.~63, no.~9, pp. 3232--3246, Sep 2015.

\bibitem{Yuan2016}
B.~Yuan and S.~A. Jafar, ``Elevated multiplexing and signal space partitioning
  in the 2 user {MIMO IC} with partial {CSIT},'' in \emph{Proc. IEEE SPAWC},
  Jul. 2016.

\end{thebibliography}

\end{document}